\documentclass[11pt]{article}

\usepackage[utf8]{inputenc}
\usepackage{lastpage}
\usepackage{fancyhdr}
\usepackage{amsmath}
\usepackage{amsthm}
\usepackage{amssymb}

\usepackage{graphicx}
\setkeys{Gin}{width=\linewidth,totalheight=\textheight,keepaspectratio}
\graphicspath{{./images/}}
\usepackage{booktabs}

\usepackage{units}

\usepackage{fancyvrb}
\fvset{fontsize=\normalsize}

\usepackage[T1]{fontenc}
\usepackage[sc]{mathpazo}
\usepackage{multicol}
\usepackage{lipsum}
\usepackage{wrapfig}
\usepackage{latexsym}
\usepackage{mathrsfs}
\usepackage{mathtools}
\usepackage{enumerate}
\usepackage{chemarrow}
\usepackage{soul}
\usepackage{orcidlink}

\usepackage{tocloft}

\usepackage{hyperref}
\hypersetup{hidelinks}

\usepackage{color,cancel}
\usepackage[normalem]{ulem}
\usepackage{verbatim}

\usepackage[doi=true,isbn=true,url=false,
citestyle=alphabetic, bibstyle=alphabetic, firstinits=true,
sorting=nyt, eprint=true, uniquename=true, uniquelist=true,defernumbers=true,backend=biber, maxnames=999]{biblatex}
\renewbibmacro{in:}{}
\DeclareUnicodeCharacter{0301}{\'{e}}

\newbibmacro{string+doi}[1]{%
  \iffieldundef{doi}{\iffieldundef{url}{#1}{\href{\thefield{url}}{#1}}}{\href{http://dx.doi.org/\thefield{doi}}{#1}}}
\DeclareFieldFormat*{title}{\usebibmacro{string+doi}{\mkbibemph{#1}}}
\DeclareFieldFormat*{title}{\usebibmacro{string+doi}{\mkbibquote{#1}}}

\AtEveryBibitem{
 \clearlist{location}
 \clearfield{month}
 \clearlist{address}
 \clearfield{date}
\clearfield{pages}
 \clearfield{issn}
}

\usepackage[nameinlink,capitalise]{cleveref}

\usepackage{xpatch}
\makeatletter
\xpatchcmd{\@cref}{\begingroup}{\begingroup\bfseries}{}{}
\makeatother

\crefformat{figure}{#2\bfseries\figurename~#1#3}
\crefformat{theorem}{#2\bfseries Theorem~#1#3}

\usepackage[labelfont=bf, font=small]{caption}

\usepackage{algorithm}
\usepackage{algpseudocodex}
\makeatletter
\algrenewcommand\ALG@beginalgorithmic{\footnotesize}
\makeatother

\newcommand{\given}{\, | \,}
\newcommand{\with}{ ; \,}
\newcommand{\suchthat}{\, : \,}
\newcommand{\f}{\varphi}
\renewcommand{\vec}[1]{\boldsymbol{#1}}
\newcommand{\domain}{\Omega}%
\newcommand{\dd}{\mathrm{d}}
\newcommand{\rbb}{\mathbb{R}}
\newcommand{\zbb}{\mathbb{Z}}
\newcommand{\nbb}{\mathbb{N}}

\newcommand{\EE}{\mathbb{E}}

\newcommand{\params}{\vec{\theta}}
\newcommand{\domainparams}{\vec{\vartheta}}
\newcommand{\data}{\vec{x}}
\newcommand{\xvec}{\vec{x}}
\newcommand{\Lcal}{\mathcal{L}}
\newcommand{\tauexit}{\tau_{\text{exit}}}
\newcommand{\tauswitch}{\tau_{\text{switch}}}
\newcommand{\ksource}{\kappa_{\mathrm{+}}}
\newcommand{\kswitch}{\kappa_{\mathrm{-}}}
\newcommand{\borel}{\mathcal{B}}

\newcommand{\ep}{\epsilon}
\newcommand{\figfont}[1]{\textsf{\textbf{#1}}}

\newtheorem{theorem}{Theorem}[section]
\newtheorem{proposition}[theorem]{Proposition}
\newtheorem{definition}[theorem]{Definition}

\newtheorem{remark}[theorem]{Remark}

\usepackage{authblk}

\title{Inferring stochastic rates from heterogeneous snapshots of particle positions}

\author[1]{\orcidlink{0000-0001-5494-403X} Christopher E. Miles}
\author[2]{\orcidlink{0000-0001-9434-9163} Scott A. McKinley}
\author[3]{\orcidlink{0000-0003-0118-5441} Fangyuan Ding}
\author[4]{Richard B. Lehoucq}

\affil[1]{Department of Mathematics, University of California, Irvine, \url{chris.miles@uci.edu}.}
\affil[2]{Department of Mathematics, Tulane University, \url{scott.mckinley@tulane.edu}.}
\affil[3]{Department of Biomedical Engineering, University of California, Irvine. }
\affil[4]{Discrete Math and Optimization, Sandia National Laboratories}

\date{November 3, 2023}

\begin{document}

\maketitle

\begin{abstract}
\noindent Many imaging techniques for biological systems -- like fixation of cells coupled with fluorescence microscopy -- provide sharp spatial resolution in reporting locations of individuals at a single moment in time but also destroy the dynamics they intend to capture. These \emph{snapshot observations} contain no information about individual trajectories, but still encode information about movement and demographic dynamics, especially when combined with a well-motivated biophysical model. The relationship between spatially evolving populations and single-moment representations of their collective locations is well-established with partial differential equations (PDEs) and their inverse problems. However, experimental data is commonly a set of locations whose number is insufficient to approximate a continuous-in-space PDE solution. Here, motivated by popular subcellular imaging data of gene expression, we embrace the stochastic nature of the data and investigate the mathematical foundations of parametrically inferring {demographic} rates from snapshots of particles undergoing birth, diffusion, and death in a nuclear or cellular domain. Toward inference, we rigorously derive a connection between individual particle paths and their presentation as a Poisson spatial process. Using this framework, we investigate the properties of the resulting inverse problem and study factors that affect quality of inference. One pervasive feature of this experimental regime is the presence of cell-to-cell heterogeneity. Rather than being a hindrance, we show that cell-to-cell geometric heterogeneity can \textit{increase} the quality of inference on dynamics for certain parameter regimes. Altogether, the results serve as a basis for more detailed investigations of subcellular spatial patterns of RNA molecules and other stochastically evolving populations that can only be observed for single instants in their time evolution.%
\end{abstract}

\linespread{0.6}         %
\tableofcontents
\linespread{1.02}         %

\section{Introduction}
\subsection{Background and Motivation}

Advances in microscopy and automated tracking have matured to an age of observing of the spatial evolution of individuals at scales ranging from tissue-scale collectives of cells \cite{cognet2014advances} down to single molecules \cite{moerner2007new,mclaughlin2020spatial}.  Such spatial tracking has led to discoveries in the understanding of viral transmission \cite{chen2014transient}, intracellular transport \cite{Osunbayo2019,rayens2023palmitate}, and many other areas in cell and molecular biology \cite{manzo2015review}. In these investigations, the dynamics of a population are probed through the tracked motion of individual trajectories.  The theory and practice of analyzing these single-particle trajectories is rich and well-developed  \cite{qian1991single,mellnik2016maximum,rehfeldt2023random}. In contrast, many imaging approaches require the fixation of cells (which requires killing the tissue and cryo-preservation to stabilize agent locations) and therefore provide only a snapshot of the population at one moment of time \cite{biswas2021imaging}. Without direct access to the evolution of creation, destruction, and motion of this population prior to imaging, what can be inferred about these underlying dynamics? 

One very active area of inferring dynamics from snapshots is in the quantification of gene expression from \textit{spatial transcriptomics} imaging of individual RNA molecules \cite{weinreb2018fundamental,marx2021method}. Although time-lapse imaging and inference of these systems is possible in certain circumstances \cite{bowlesScalableInferenceTranscriptional2022}, inferring the dynamics underlying populations from RNA snapshots is a far wider and mainstream interest. Such spatial transcriptomics have served to be invaluable in recent identifications of cell types \cite{moses_museum_2022} and disease mechanisms \cite{williams_introduction_2022}. However, the theory to analyze this data has lagged behind the sophisticated techniques used to harvest it. Spatial imaging techniques (such as smFISH, single-molecule FISH \cite{shaffer2013turbo,omerzuThreedimensionalAnalysisSingle2019,Ding2019}) are capable of resolving individual RNA molecules, multiplexed over several genes \cite{maynardDotdotdotAutomatedApproach2020} and many cells\cite{cai_stochastic_2006,eng_transcriptome-scale_2019}.  This unprecedented single-molecule resolution is largely neglected, with most analysis approaches focusing on RNA counts, binned either per-cell or split into nuclear and cytoplasmic counts \cite{la2018rna,svensson2018rna,gorin2022rna,fu2022quantifying}. In doing so, the subcellular spatial factors that crucially control gene expression \cite{heinrich_temporal_2017,cassella_subcellular_2022}, e.g., geometry-dependent nuclear export \cite{kohler2007exporting,rodriguez2011linking}, are largely understudied and unincorporated into transcriptomics analyses.

Beyond the motivation for quantification of  subcellular patterns of individual RNA molecules, we also emphasize the goals and value of \textit{mechanistic} modeling. A zoo of sophisticated techniques has arisen for the analysis of (spatial) transcriptomics data, but these largely phenomenological or statistical studies \cite{imbertFISHquantV2Scalable2022,mah2022bento} suffer a lack of methodological {reproducibility} and interpretability \cite{gorin2022interpretable,chari2023specious}. The proposed solution to these shortcomings is the use of mechanistic, stochastic models \cite{lammers2020matter,gorin2022interpretable} as the basis of inferring gene expression dynamics. Such methods have grown in popularity and sophistication in recent years \cite{munsky2015integrating,gomez2017bayfish,herbach2017inferring,,gorin2022modeling,luo2022bisc,bowlesScalableInferenceTranscriptional2022,guptaInferringGeneRegulation2022,kilic2023gene,vo_analysis_2023}, but remain primarily focused on non-spatial models. Thus, we have outlined the motivation for subceullar, mechanistic modeling of RNA molecules as a natural next step in expanding the frontier of spatial transcriptomics quantifications. We anticipate such techniques will be useful in advancing understanding of subcellular, spatial factors that modulate gene expression, including post-transcriptional splicing \cite{Ding2019,Cote2020Biorxiv,ding2022dynamics} and nuclear export \cite{kohler2007exporting,rodriguez2011linking}.  %

The basis of this work is a mathematical connection between spatial point processes and snapshot observations of stochastic birth-death-diffusion processes. Most crucially, this connection provides statistical information of stochastic particle positions and counts, while circumventing the need for costly spatial stochastic simulation \cite{lemerle2005space}. %
The connection was suggested by Gardiner in the 1970s \cite{gardiner1977poisson,chaturvedi1978poisson,chaturvedi1977stochastic} and recently reintroduced non-rigorously in \cite{schnoerr2016cox} in broad generality of stochastic reaction-diffusion processes. %
Here, we provide rigorous particle-perspective justification for this connection in the setup of birth, death, and diffusion in a domain. The machinery developed to prove this reveals new insight into generalizations and limitations of the theory.  Inference on spatial point processes is a well-established topic in the statistics literature \cite{baddeley_likelihoods_2001}, but focuses largely on phenomenological or purely statistical modeling of the underlying intensities, including in a recent application to the same data considered as our work \cite{walter2023fishfactor}. In contrast, the resulting setup here is an intensity derived from the solution of a birth-diffusion-death PDE. The model problem is similar in structure to classical PDE-based "inverse problems" \cite{aster2018parameter,flegg2020parameter} but the crucial difference is that the data presents as collections of stochastic particle locations rather than the PDE solution itself plus observational noise. Here we embrace the stochastic particle-by-particle nature of the data and study the system using spatial point process inference techniques, which are naturally informed by the solutions of appropriate PDEs.%

\subsection{Overview of work}

\begin{figure}
    \centering
    \includegraphics[width=\textwidth]{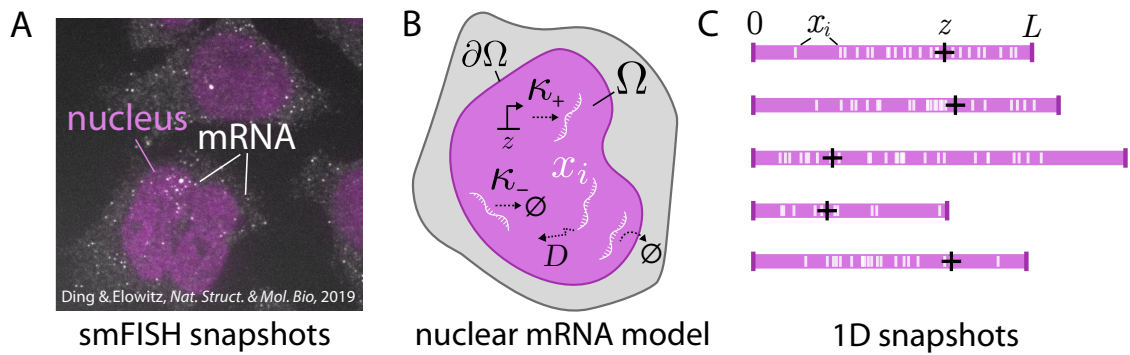}
    \caption{\textbf{Motivation and setup of the work.} \figfont{A}: Experimental FISH data from \cite{Ding2019} demonstrating the basis for modeling within this work. Individual mRNAs are resolved in both the nucleus and cytoplasm. \figfont{B}: Schematic of the mRNA model in nuclear domain $\Omega$: birth at rate $\ksource$ at gene site $z$, decay at rate $\kswitch$, and diffusion with diffusivity $D$, and instantaneous nuclear export at the boundary $\partial \Omega$ (an absorbing boundary). \figfont{C}: 1D simulations of the model demonstrating the heterogeneity in nuclear (or cellular) domain size $L$ and gene site $z$ (black plus).   }
    \label{fig:datafig}
\end{figure}

In this work, we investigate a "toy" model motivated by smFISH imaging data of subcellular RNA \cite{Ding2019} that distills salient features of the intrinsically stochastic \cite{raj2006stochastic,raj2008nature}, spatial \cite{heinrich_temporal_2017, cassella_subcellular_2022}, and heterogeneous nature of the data. In such imaging, such as the example shown in \cref{fig:datafig}\figfont{A}, individual mRNA can be resolved in both cytoplasm and the nucleus of cells. Panel \figfont{B} shows the model considered throughout the remainder of the text for nuclear mRNA in a domain $\domain$, with stochastic production at rate $\ksource$ at a gene site $z$, decay at rate $\kswitch$, and diffusion with diffusivity $D$. The nuclear boundary is assumed to be instantaneous export at $\partial \domain$. To faithfully reflect the cell-to-cell heterogeneity that crucially affects the observed data, we incorporate different nuclear (or cellular) domain sizes and gene site positions, shown through a simplified 1D setup in \figfont{C}. The remainder of this work explores the mathematical nature of the inverse problem on this collection of heterogeneous snapshots. 

In \cref{sec:methods} we outline the model setup of snapshot data as an observation of a distribution of particles at time zero in a system in which the initial condition has been pulled back in time to negative infinity.  We show that the snapshot data can be described as a nonhomogeneous spatial Poisson point process with an intensity measure that is generally recognized as the  \emph{occupation measure} of diffusing particles before exiting the domain or being removed by state-switching \cite{durrett1996stochastic}. The full proof is provided in \cref{sec:support:thm} with a  nondimensionalization argument found in \cref{sec:support:nondim}. This connection is then translated into a likelihood function that can then be used for inference of parameters. 

In \cref{sect:1d_model}, we put the theory to practice in the context of our toy model. In this setting, the likelihood function has a fully explicit form, which we report and visualize as a function of the number of cell snapshots that are included in a given experiment. This raises questions concerning the impact on inference from geometric heterogeneity among the cells %
and, in \cref{subsect:like_and_info}, we calculate the (Fisher) information for experiments run under different heterogeneity assumptions. In \cref{subsect:inference}, we explore how these theoretical results translate into uncertainty in the inference of model parameters. We report on numerical experiments in which we generated synthetic snapshot data from simulations of underlying particle processes, and employed Bayesian methods to quantify joint uncertainty in the inference of particle production and degradation rates. In \cref{subsect:hetero} we report on similarly constructed numerical experiments based on several different cell-heterogeneity assumptions. We find that the statistical information analysis conducted in \cref{subsect:inference} does faithfully predict Bayesian uncertainty quantifications that result from simulated data. This encourages continued use of these methods for both experimental design and practical statistical inference. %
Finally, in \cref{sec:supp:infobin} we use the information-theoretic framework to demonstrate the benefit of using explicit particle locations for inference when compared to spatial binning. Additional numerical implementation details regarding stochastic process simulations and MCMC sampling for Bayesian inference are outlined in \cref{sect:numerics}.

\section{Mathematical and Statistical Methods}
\label{sec:methods}

\subsection{Notation}

Throughout this study, we have it in mind that biological cells can be heterogeneous in their geometry, but have similar internal reaction rates and transport dynamics. Let $M \in \mathbb{N}$ be the number of cells in a study, and for each $m \in \{1, 2, \ldots, M\}$, let the compact set $\domain_m \subset \rbb^d$ denote the domain of the $m$th cell. Within each cell there are diffusive particles that emerge at rate $\ksource$ at locations that are distributed according to a source probability distribution $\f_m : \Omega_m \to \rbb_+$. They then move according to the laws of Brownian motion with uniform diffusivity $D$, possibly with drift $\alpha : \domain \to \rbb^d$. If a particle in the $m$th cell exits the domain $\domain_m$, we assume that it cannot return. Furthermore, at rate $\kswitch \geq 0$, particles undergo a change of state (either they degrade, or engage in a chemical reaction, for example). We will call a particle \emph{active} up until the time it either switches state or exits the domain, and \emph{inactive} thereafter.

In each cell, at time $t = 0$ we observe the locations of the active particles, denoted $\vec{x}_m = \{x_{1m}, x_{2m} \ldots, x_{n(\vec{x}_m),m}\}$, where $x_{im} \in \domain_m$. We write $n(\vec{x})$ for the number of particles in the collection $\vec{x}$, and emphasize that the number of particles in each of the $M$ distinct cells $\domain_m$ will be independent of each other and random.
In the foregoing discussion, we emphasize that the boundary condition, the diffusivity constant, and the emergence and switch rates are universal among all cells, but the initial location distribution (often realistically in the cell nucleus) can be different in each different cell. 

As we will see later, the three parameters $D$, 
$\ksource$ and $\kswitch$ are not mutually identifiable. We therefore introduce the ratios
\begin{equation} \label{eq:defn-lambda-mu}
    \lambda \coloneqq \frac{\ksource L_0^2}{D}, \text{ and } \mu \coloneqq \frac{\kswitch L_0^2}{D}.
\end{equation}
That is, we are nondimensionalizing with respect to the particle diffusivity and the typical cell size $L_0$ over the population of observations. 

\subsection{Stochastic model for snapshot observations within a single cell}
\label{methods:theorem}

We now formally introduce the stochastic model that provides the basis for our inference protocol. For the remainder of this section, we will suppress dependence on the cell index $m$ and describe the dynamics within a single focal cell.

In this work, we assume that the emergence (or birth) times of particles are given by a stationary Poisson point process on the negative half of the real line. 
Let $0 > T_1 > T_2 > \dots$ be an enumeration of points in the negative half-line, starting with the point closest to zero and decreasing from there. The initial locations of particles will be given by the sequence of independent and identically distributed (iid) random variables $\xi_i \overset{iid}{\sim} \f$. Let $\{\mathring{B}_i(t)\}_{t \geq 0}$ be a sequence of iid standard $d$-dimensional Brownian motions (with $\mathring{B}_i(0) = 0$), and let $B_i(t) \coloneqq \mathring{B}_i(t-T_i)$ be the Brownian motions shifted so that their initial time is the associated particle emergence time $T_i$. 

We use the term ``time'' in the preceding description, but as we show in \cref{sec:support:nondim},%
 we can assume that the dynamics have already been rescaled to nondimensional coordinates that result in a process that has diffusivity one, switch rate $\mu$ and birth rate $\lambda$. In this nondimensional framework, let a drift vector $\alpha \colon \rbb^d \to \rbb^d$ be given, and define $\{X_i(t)\}_{t > T_i}$ to be the solution of the stochastic differential equation (SDE)
\begin{equation} \label{eq:defn-X}
\begin{aligned}
    \dd X_i(t) &= \alpha(X_i(t)) \dd t + \sqrt{2} \, \dd B(t), \quad t > T_i; \\
    X_i(T_i) &= \xi_i.
\end{aligned}
\end{equation}
The partial differential operator associated with this diffusion is defined in terms of its action on twice-differentiable functions $f \colon \rbb^d \to \rbb$, 
\begin{equation} \label{eq:defn-L}
    \Lcal_X \coloneqq \alpha(x) \cdot \nabla f(x) + \Delta f(x)
\end{equation} where $\Delta$ is the Laplacian. 

Associated with each path is a continuous-time Markov chain (CTMC) $\{J_i(t)\}_{t > T_i}$ that records the state of the particle at time $t$. We assume that all particles start in the \emph{active} state $J_i(0) = 1$. At rate $\mu$, the particles switch to the inactive state 0. We label the switch times $\tau_{\mathrm{switch},i}$ and, due to the Markov assumption, the switch times are exponentially distributed with densities $p_i(t) = \mu e^{-\mu (t - T_i)} 1_{[T_i, \infty)} (t)$. Because particles are assumed to not be able to return to the interior of a cell once they exit, we define the stopping times $\tau_{\mathrm{exit},i} \coloneqq \inf\{t \geq T_i \suchthat X_i(t) \notin \domain\}$.

\begin{figure}[htb]
    \centering
        \includegraphics[width=\textwidth]{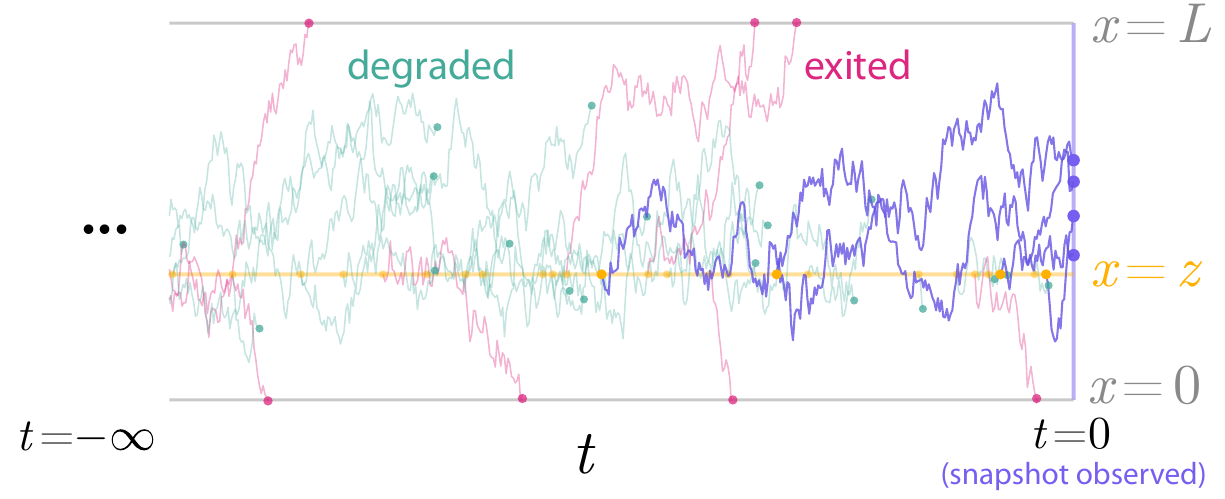}
    \caption{\textbf{Visualization of the stochastic model.} Particle populations observed at the snapshot instant ($t = 0$) have long histories that cannot be observed. Here, particles are ``born'' at a rate $\lambda$ (yellow circles) in the source location $z$ (yellow line). The particles diffuse until  they are either degraded (green circles) or exit the domain (pink circles). The faded trajectories are for particles that left the system before time zero, while the purple trajectories are those that remain active and are observed in the $t=0$ snapshot. Nondimensional $x,t$ shown with rates $\lambda=10$, $\mu=1.5$ corresponding original simulation parameters $\ksource=5, \kswitch=0.75, D=0.5, L=3,$ $L_0=1, z=1$.    }\label{fig:snapshot-roots}
\end{figure}

Our model for the ``snapshot'' data is to define a spatial point process $\vec{N} \colon \borel(\domain) \to \zbb$ that counts the number of \emph{active} particles in subsets of the domain at time 0. To be precise, for every Borel subset $A \in \borel(\domain)$, we define
\begin{equation} \label{eq:defn-N}
    N(A) \coloneqq \sum_{i = 1}^\infty \delta_{A} \big((X_i(0 \wedge \tau_{\mathrm{exit},i})\big) 1_{\{J(0 \wedge \tau_{\mathrm{exit},i}) = 1\}},
\end{equation}
where $\delta_A(x) = \int_A \delta(x-y) \dd y$ and $t \wedge s \coloneqq \min(t,s)$. In other words, $N(A)$ counts all particles that have neither exited the system through the boundary $\partial \domain$, nor switched from active to inactive, and are located in the set $A$ at time $t=0$. A schematic of this can observation model can be seen in \cref{fig:snapshot-roots}.

We can now state the main mathematical result, which identifies the relationship between collections of particle locations in steady state and the solution of a well-known boundary value problem (BVP).

\begin{theorem} \label{thm:main}
Let ${(X_i(t), J(t))}_{i \in \nbb}$ be the particle/state pairs defined by \eqref{eq:defn-X} and the foregoing discussion. Let $\domain \subset \rbb^d$ be an open domain with compact closure that has a boundary with all regular points with respect to the SDE \eqref{eq:defn-X}. Then the associated spatial point process $\vec{N}$, defined by \eqref{eq:defn-N}, giving the locations of \emph{active} particles at time zero, is a Poisson spatial process with intensity measure $u$ satisfying the BVP 
\begin{equation} \label{eq:defn-u}
\begin{aligned}
    \Lcal^* u(x) - \mu u(x) &= -\lambda \f(x), & x \in \domain, \\
    u(x) &= 0,  & x \in \partial \domain.
\end{aligned}
\end{equation}
where $\Lcal^* u = -\nabla \cdot (\alpha(x) u(x)) + \Delta u(x)$ is the adjoint of the operator \eqref{eq:defn-L}.
\end{theorem}

\begin{remark} \label{rem:poisson}
This is to say, $\vec{N}$ satisfies the condition that, given any disjoint sets $A_1, \ldots A_K \subset \borel(\domain)$, the random variables $\{N(A_k)\}_{k = 1}^K$ are independent with respective distributions
\begin{equation}
\label{eq:Npoiss}
    N(A_k) \sim \mathrm{Pois}\Big(\int_{A_k} \!\! u(x) \dd x \Big).
\end{equation}
\end{remark}
The proof of this theorem is the content of \cref{sec:support:thm}. 

\subsection{Framework for inference}

The Poisson spatial point process representation of the particle locations organizes the data into a form that readily allows for statistical inference. Suppose that we partition the domain $\domain$ into a fixed collection of sets. We then ``bin the data,'' meaning we record the particle locations in terms of the number of particles in each bin. By \cref{thm:main}, the number of particles in each bin is Poisson distributed and independent of number of particles in other bins. This yields a likelihood function that can be used for parameter inference, but depends explicitly on the choice of bins. For example, suppose that the sets $(A_1, A_2, \ldots, A_K)$ form a partition of cell domain, meaning that they are mutually disjoint and their union is $\domain$. For a given set of particle locations $\xvec$, define $n_k(\vec{x})$ to be the number of particles observed in the set $A_k$. Then the likelihood function for the parameter vector $\theta = (\lambda, \mu)$ given the binned data is
\begin{equation} \label{eq:likelihood-bins}
\begin{aligned}
    P_\theta\big(N(A_1) = n_1(\xvec), &\ldots, N(A_K) = n_K(\xvec)) \\
    &= \prod_{k = 1}^K e^{-\int_{A_k} \!\! u(x \with \theta) \dd x} \Big(\int_{A_k} \!\! u(x \with \theta) \dd x\Big)^{n_k(\vec{x})} \frac{1}{n_k(\vec{x})!}.
\end{aligned}
\end{equation}
We introduce the notation $u(x\with\theta)$ to emphasize the dependence of the intensity measure on $\theta$ through the BVP \eqref{eq:defn-u}.

Now, if we change the partition of the domain, the associated likelihood function for the new binning of the data will be different, leading to different estimates and uncertainty regions for the parameters. Naturally among all choices for which partition to use, we would prefer to use one that yields the ``best'' statistical inference in some rigorous sense. One standard approach is to consider the likelihood that emerges in a limit in which the partition mesh size (often defined to be the size of the largest bin) is taken to zero. The form of the limiting likelihood is well known for spatial Poisson processes \cite{baddeley_likelihoods_2001, daley2003introduction} and can be expressed directly in terms of the particle locations, rather than a summary from binning.

\begin{definition}[Likelihood function] 
\label{defn:likelihood}
Let $\domain \subset \rbb^d$ satisfy the conditions of \cref{thm:main}. For each $\params = (\lambda, \mu) \in \Theta$, let $u(\cdot \with \params)$ satisfy the BVP \eqref{eq:defn-u} for the given parameter vector. Then for any set of particle locations $\vec{x} = \big\{x_1, x_2, \ldots, x_{n(\vec{x})}\big\} \subset \domain$, we define the likelihood function $L(\cdot \with \vec{x}) \colon \Theta \to \rbb_+$ as follows:
\begin{equation} \label{eq:likelihood}
    L(\params \with \vec{x}) \coloneqq \left( \prod_{i=1}^{n(\vec{x})} u(x_i \with \params) \right) e^{- \int_\domain u(x \with \params) \dd x} .
\end{equation}
\end{definition}

While the use of this form of likelihood function is standard, in \cref{sec:support:likelihood} we include a derivation using the Lebesgue Differentiation Theorem in the context of Bayesian methods that we have not explicitly seen in the literature. Moreover, in \cref{sec:supp:infobin}, we use the notion of statistical information to quantify the claim that the likelihood function defined by \eqref{eq:likelihood} is ``better'' than any likelihood function of the form \eqref{eq:likelihood-bins} that arises from binning the particle locations. 

We note that the above discussion contains a mix of Bayesian and frequentist ideas and we will continue to draw on tools from each school throughout this work. Generally speaking, when we are considering a specific data set and wish to establish point estimates for the parameters and to quantify uncertainty, we will use a Bayesian method.  This is because the notion of a highest-density posterior region (HDPR) allows for clean visualization of 2d ``joint uncertainty'' for the parameters $\lambda$ and $\mu$, which can then be summarized with a single quantity: its area. Meanwhile, when we wish to make comparisons that are not specific to a given data set, we will use the determinant of the information matrix $\det(I)$. Optimization of this determinant, or \textit{D-optimality} \cite{atkinson2007optimum}, is a popular choice in experimental design (including in other studies of gene expression \cite{fox2019finite}) due to its ability to capture correlations between inferred parameters in the (inverse) area of confidence ellipsoids. However, confidence is a frequentist notion, and we use $\det(I)$ as an analytically tractable proxy for the inverse of the area of the joint uncertainty region in Bayesian inference. The relationship between maximizing $\det(I)$ and minimizing expected HDPR is not precise for finite data, in part because the role of the prior must be considered in the Bayesian setting \cite{gabor2015robust,sharp2022parameter}. But, as we show in \cref{subsect:hetero}, $\det(I)$ does make useful predictions in the present setting.

\section{Inferring dynamics of the 1D RNA model}
\label{sect:1d_model}

Using the framework established in previous sections,  we now explore their consequences and implementation in the context of the 1D model approximating the key features seen in subcellular spatial RNA transcriptomics data. The principal result was the connection between the PDE \eqref{eq:defn-u} as the intensity for the corresponding point process. For a single cell, this PDE becomes  
\begin{equation}\label{eq:u_specific}
0 =  \Delta u(x)  +  \lambda \delta(x-z) - \mu u(x)  \,\text{ for } x \in [0,L], \quad u(0) = 0, \, u(L)=0,
\end{equation}
where $(\lambda,\mu)$ are nondimensional \eqref{eq:defn-lambda-mu} dynamic quantities sought to be inferred, and $L$ is relative to the typical length seen in the population of images.  

The solution to \eqref{eq:u_specific} can be explicitly computed:

\begin{equation}
\label{eq:usol}
u(x) =
 \frac{\lambda}{\sqrt{\mu}}  \operatorname{csch}\left(\sqrt{\mu } L\right) \sinh \left(\sqrt{\mu } \min\{x,z\} \right) \sinh \left(\sqrt{\mu } (L-\max\{x,z\} )\right).
\end{equation}
The expected number of particles is then, from \eqref{eq:Npoiss},
\begin{equation} 
\label{eq:eNsol}
\EE[N] = \int_{0}^L u(x) \, \mathrm{d} x = \frac{\lambda}{\mu} \left(1-\text{sech}\left(\frac{\sqrt{\mu } L}{2}\right) \cosh \left(\frac{1}{2} \sqrt{\mu } (L-2 z)\right) \right) 
\end{equation}
The expectation in \eqref{eq:eNsol} shows an intuitive structure of the product of the pure birth-death mean $\lambda/\mu$, and a geometry-dependent term (always strictly less than 1) that scales the former. 

\subsection{Likelihood and information from analytical solution}
\label{subsect:like_and_info}
Equipped with the statistical characterization from \cref{thm:main} and the analytical expression of the intensity for a single snapshot in \eqref{eq:usol}, we can now examine the behavior  and emergent lessons arising from this inference setup.

\begin{figure}[htb]
    \centering
        \includegraphics[width=1\textwidth]{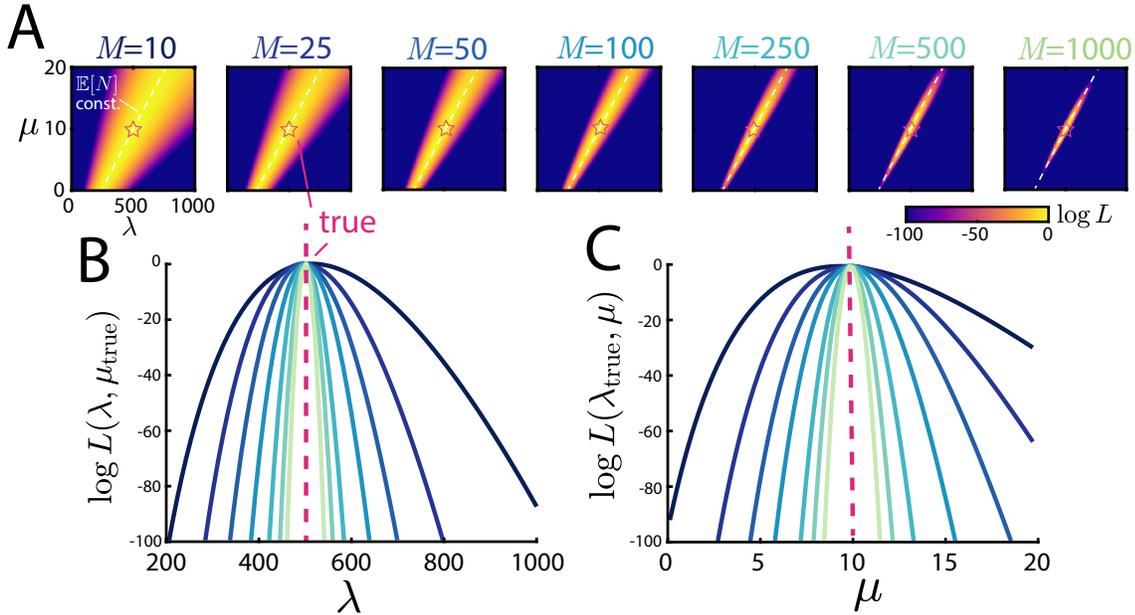}
    \caption{\textbf{Behavior of the likelihood.} \figfont{A}: $\lambda-\mu$ evaluations of the likelihood $L(\params \with \data)$ for increasing numbers of snapshots $M$ with fixed $\lambda=500, \mu =10, L=1, z=0.5$. At low $M$, the likelihood peaks around the line of constant $\EE[N]$, and as $M$ increases, sharply peaks around the true parameter values. \figfont{B/C}: Slices of the likelihood shown in panel \figfont{A} evaluated at $\mu_{\text{true}}$ and $\lambda_{\text{true}}$ respectively. }
    \label{fig:logL}
\end{figure}

We first explore the behavior of the likelihood \eqref{eq:likelihood} as a function of the system parameters and quantity of data. To do so, we perform stochastic simulations of the underlying birth-death-diffusion trajectories and then terminate them at an experimental time $T$. The particle positions at this observation time are the data used for inference. Further details on the stochastic simulation can be found in \cref{sect:numerics}. Beginning with the simplest setup, we fix $\lambda=500, \mu=10$, $L=1$ and $z=1/2$ and explore how increasing the number of snapshots $M$ shapes the likelihood, the results of which can be seen in \cref{fig:logL}\figfont{A}. As $M$ increases, the likelihood first concentrates around a line and then within this line, begins to peak at the true parameter values. To understand this behavior, we plot the $\lambda-\mu$ line of constant $\mathbb{E}[N]$ from \eqref{eq:eNsol} superimposed on the likelihood. From this, we see that the likelihood concentrates around the line defining the true number of particles, and then further resolves within this line from their spatial information. In panels \figfont{B} and \figfont{C} of the same figure, the same information is shown as slices of the likelihood with one parameter at a time held to be the true value. From these panels, we see again that sufficiently large $M$ (here, $M\approx 1000$) yields likelihoods that are tightly concentrated around the true values used to generate the synthetic data. In the next section, we will explore more concretely how the spread of the likelihood translates into uncertainty in the inference itself. 

Before performing inference on the synthetic data, other useful quantities can be computed directly from the likelihood. For a likelihood $L(x;\vec{\theta})$, the corresponding information matrix \cite{vahid2020fisher} is defined by the elements 
\begin{equation}
I_{ij} \coloneqq \EE_{\params} \left[ (\partial_{\theta_i} \log L(\params \with \data))
    (\partial_{\theta_j} \log L(\params \with \data))\right] = -\EE_\theta \left[ \partial_{\theta_i}\partial_{\theta_j} L(\params \with \data)\right],
    \end{equation}
    where the latter equality holds under regularity conditions straightforward to verify for Poisson point processes \cite{clark2022cramer}. Our parameters of interest are  $\vec{\theta}=(\lambda,\mu)$ so the information matrix takes the form
\begin{equation} \label{eq:infomat_2x2}
    I = -\mathbb{E}_\theta \begin{bmatrix} \partial_{\lambda \lambda} L(x;\lambda, \mu)& \partial_{\lambda \mu} L(x;\lambda, \mu) \\ \partial_{\lambda \mu} L(x;\lambda , \mu) & \partial_{\mu \mu} L(x;\lambda, \mu)\end{bmatrix}.
\end{equation}
Importantly, the expectation is over both stochastic particle positions and number. To do so, we employ machinery from the theory of spatial point processes, specifically the Campbell-Hardy formula stated in \cref{thm:campbell}.
\begin{theorem}[Campbell-Hardy, \cite{baddeley_likelihoods_2001}]\label{thm:campbell}
Let $N$ be governed by a general point process on domain $x\in\domain$ with intensity $u(x)$ so that $\EE[N(B)] = \int_B u(x) \dd x$.Let  $f(x)$ be any measurable function. Then, the expectation of the sum over observations of the point process is 
\begin{equation}
  \mathbb{E}\left [  \sum_{i=1}^{N} f(x_i) \right] = \int_{\domain} f(x) u(x) \dd x. \label{eq:campbell}
\end{equation}
\end{theorem}
Through \cref{thm:campbell}, the elements of our information matrix in \eqref{eq:infomat_2x2} can be related to the PDE solution $u$ by \eqref{eq:usol}. To do so, the log-likelihood $\ell(\params;\data)$ from \eqref{eq:likelihood} is $$
\ell(\params;\data) \coloneqq \log L(\params \with \data) = -\int_\domain u(x;\params) \dd x + \sum_{i=1}^{N} \log u(x_i;\params)$$ and $\partial_{\theta_i \theta_j} \ell = -\int_\domain \partial_{\theta_i \theta_j} u(x;\params) \dd x + \sum_{i=1}^{N} \partial_{\theta_i\theta_j}\log u(x_i;\params) $.
\begin{equation}
    \begin{aligned}
        I_{ij}(\params)  =-\EE[\partial_{\theta_i\theta_j} \ell(\params;\data)] &= -\EE \left [-\int_\domain \partial_{\theta_i \theta_j} u \dd x \right] - \EE\left[\sum_{i=1}^{N} \partial_{\theta_i\theta_j}\log u(x_i;\params) \right]\\ 
        &= \int_\domain \partial_{\theta_i \theta_j}  u(x;\params) \dd x - \EE\left[\sum_{i=1}^{N} \partial_{\theta_i\theta_j}\log u(x_i;\params) \right] \\
        &=  \int_\domain \partial_{\theta_i\theta_j} u(x;\params)  \dd x - \int_\domain (\partial_{\theta_i\theta_j} \log u (x;\params )u(x;\params) \, \dd x. \label{eq:I_ij_intermed}
        \end{aligned}
        \end{equation}
        The first term is constant, and therefore is its own expected value. The second, stochastic sum, takes the form of \eqref{eq:campbell} with $f(x)=\partial_{\theta_i \theta_j} \log u(x)$ and is evaluated with \cref{thm:campbell}. The integrands of \eqref{eq:I_ij_intermed} can be manipulated to a more tractable form by expanding derivatives
        \begin{equation}
            \begin{aligned} 
                I_{ij}(\params) &=  - \int_\domain (\partial_{\theta_i\theta_j} \log u(x;\params)  )u(x;\params) \, \dd x + \int_\domain \partial_{\theta_i\theta_j} u(x;\params)  \, \dd x\\   
        &= \int_\domain (\partial_{\theta_i}\log u) (\partial_{\theta_{j}}\log u(x;\params) ) u(x;\params)  \,  \dd x.\\ 
        &= \int_\domain \frac{(\partial_{\theta_i} u(x;\params) ) (\partial_{\theta_j} u(x;\params) )}{u(x;\params) } \,  \dd x. \label{eq:Info}
    \end{aligned}
\end{equation}
If $\lambda>0$, then $u>0$ and therefore $I_{ij}>0$ so identifiability \cite{WIELAND202160,browning2020identifiability} of both parameters can be concluded.  Moreover, the analytical solution to the PDE allows us to compute the information explicitly. For instance, 
\begin{equation}
I_{11}(\params) = -\EE_\theta\left[ \partial_{\lambda \lambda} \ell(\vec{\theta};\data) \right] =  \left(\lambda  \mu\right)^{-1} \left[ 1-\operatorname{sech}\left(\sqrt{\mu } L/2\right) \cosh \left( \sqrt{\mu } (L-2 z)/2\right) \right]. 
\end{equation}
The remaining entries can be computed straightforwardly using \texttt{Mathematica} but have unwieldy expressions not shown here.

\begin{figure}[htb]
    \centering
    \includegraphics[width=0.8\textwidth]{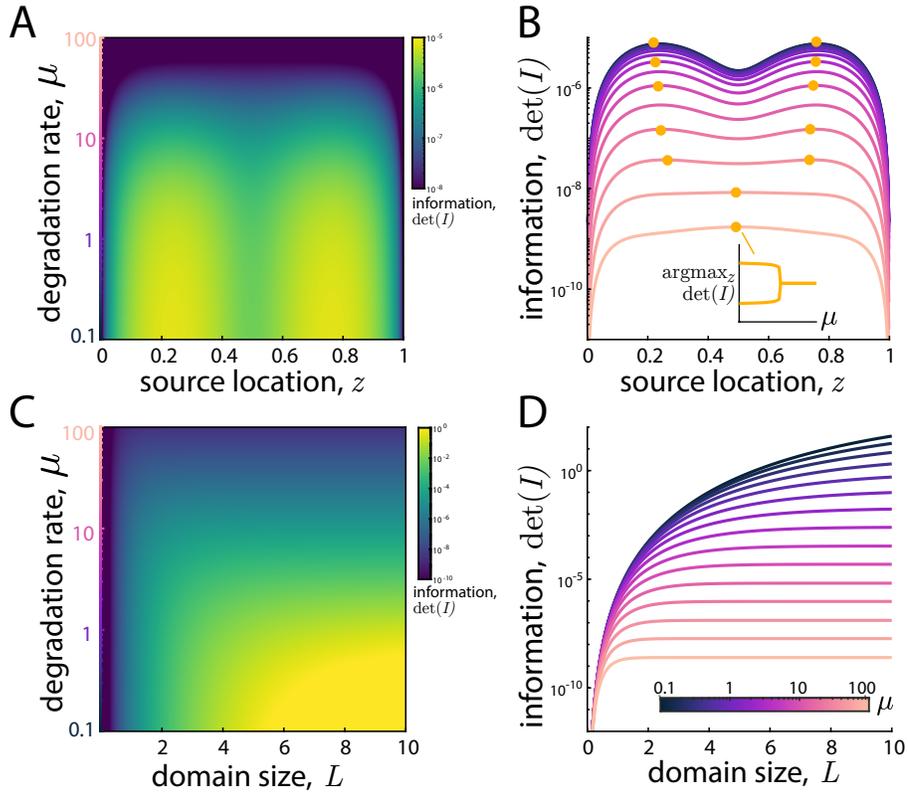}
    \caption{\textbf{Behavior of the information.} \figfont{A}: Determinant of the information $\det I$ in the $z-\mu$ plane with $\lambda=500$, $L=1$ fixed. \figfont{B}: same information as panel \figfont{B}, with $\det I$ as a function of $z$ shown for various values of $\mu$, with the maximizing value of $z$ labeled. \figfont{B inset}: the-information maximizing source location $z$ as a function of $\mu$. \figfont{C}: Determinant of the information $\det I$ in the $L-\mu$ plane with $\lambda=500$, $z=L/2$ fixed. \figfont{D}: same information as panel \figfont{C}, with $\det I$ as a function of $L$ shown for various values of $\mu$.  }
    \label{fig:info}
\end{figure}

In \cref{fig:info}, we employ the analytical expressions for the information matrix to explore the behavior of $\det I$ as a function of the underlying parameters. We interpret this quantity's importance through the classical lens of the Cram\'{e}r-Rao bound \cite{clark2022cramer}: any unbiased estimator of the parameter's variance is bounded below by the inverse of information. Informally, we interpret larger information setups as "easier" inverse problems. We later verify this interpretation through posterior widths on actual inference setups. 

In \cref{fig:info}\figfont{A} and \figfont{B}, we fix $\lambda=500, L=1$ and vary both source location $z$ and degradation rate $\mu$. For large values of $\mu$, the behavior is intuitive: the information is maximized when the source location is at the center, as this minimizes the overall number of particles exiting the boundary and corresponding information loss. For smaller values of $\mu$, a bifurcation in the behavior occurs and the optimal source location becomes non-trivial: occurring at two symmetric locations somewhere between the cell center and boundary. This non-trivial behavior arises only in the simultaneous inference of parameters, suggesting its origin. We interpret this non-trivial behavior as arising from the trade-off between seeing more particles (sources closer to the center) or seeing a wider range of the spatial gradient (for off-center source locations). 

Panels \figfont{C} and \figfont{D} investigate the cell size's impact on the information by fixing $z/L$ and varying $L$ and $\mu$. Intuitively, larger cells (relative to the typical size) yield more information due to less information loss through the absorbing boundary. Interestingly, this effect saturates for sufficiently large values of $\mu$, when the typical lengthscale of a trajectory is far shorter than the domain size, so no additional information is gained through enlarging the domain.

\subsection{Bayesian inference on synthetic data}

\label{subsect:inference}

For its noted strengths in robustness and natural uncertainty quantification that we defer to arguments elsewhere \cite{yau2019bayesian}, we will deploy Bayesian inference. We must therefore specify prior distributions \cite{gelman1995bayesian} for $\lambda$ and $\mu$.  For the production rate $\lambda$, we take
\begin{equation}
   \lambda \sim \mathrm{Gamma}(\alpha,\beta), \qquad p(\lambda) = \frac{\lambda^{\alpha-1}e^{-\beta \lambda} \beta^\alpha}{\Gamma(\alpha)}.
\end{equation}
Assuming the total prior is independent for each parameter, $p(\params) = p(\lambda)p(\mu)$, then for a single observation, we have
\begin{equation}
\begin{aligned}
    p(\params|x) &\propto  p(\params)  L(\params ; \data) 
    = p(\mu) p(\lambda) e^{-\lambda \int_\domain v(x;\mu) \dd x} \prod_{i=1}^{N}\lambda v(x_i;\mu) \\
    &= p(\mu) \prod_{i=1}^{N}v(x_i;\mu) \frac{\lambda^{\alpha-1}e^{-\beta \lambda} \beta^\alpha}{\Gamma(\alpha)} \lambda^n e^{-\lambda \int_\domain v(x;\mu) \dd x}\\
    &= p(\mu) \prod_{i=1}^{N}v(x_i;\mu) \frac{\lambda^{(\alpha+n)-1}e^{-\lambda \left(\beta + \int_\domain v(x;\mu)\dd x\right)}  \beta^\alpha}{\Gamma(\alpha)}. 
\end{aligned}
\end{equation}
Here, we have introduced the definition of $v$ as $u(x;\mu,\lambda)=\lambda v(x;\mu)$, noting that \eqref{eq:usol} is linear in $\lambda$ and following standard practice of identifying a "scale" parameter in the point process literature \cite{baddeley_likelihoods_2001}. This yields the convenient conditional posterior for $\lambda$ arising from conjugacy
\begin{equation}
    p(\lambda| \data, \mu) = \mathrm{Gamma}\Big(\alpha + n(\data), \beta + \int_\domain v(x;\mu) \, \dd x\Big).
\end{equation}
For $M$ observations, each with a different domain, this generalizes to
\begin{equation}
\label{eq:gamma_conj_update}
    p(\lambda| \data_1, \ldots, \data_m, \mu) = \mathrm{Gamma}\left(\alpha + \sum_{m=1}^{M}n(\data_m), \beta + \sum_{m=1}^{M} \int_{\Omega_m} v(x;\mu) \, \dd x\right).
\end{equation}
For $\mu$, we choose a uniform prior $\mu \sim \mathrm{unif}(0,\mu_\mathrm{max})$ with large $\mu_\mathrm{max}$ for uninformedness. With these choices, of priors, it is straightforward to implement a Gibbs-in-Metropolis sampler for the posterior distribution, with a Gibbs step using the conjugacy for $\lambda$ and Metropolis step for $\mu$. Further details on the MCMC implementation (including verification of convergence) can be found in \cref{sect:numerics}. 

\begin{figure}[htb]
    \centering
    \includegraphics[width=0.7\textwidth]{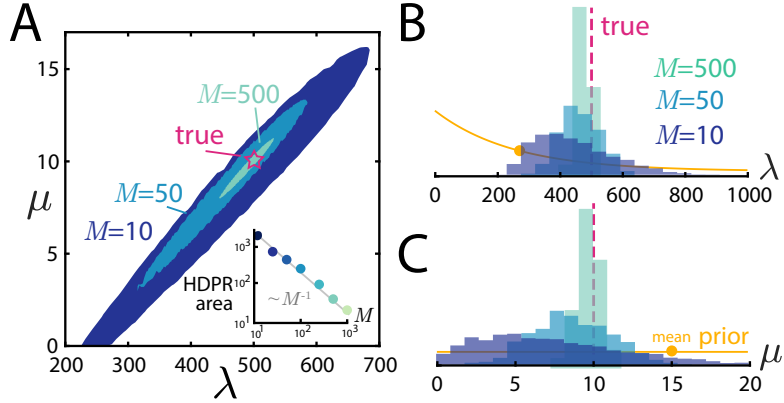}
    \caption{\textbf{Bayesian inference on synthetic data.} \figfont{A}:  Shaded regions show 89\% highest-density posterior region (HDPR) for $M=10, 50, 500$ snapshots and fixed parameters $\lambda=500, \mu=10, z=1/2,L=1$. As $M$ increases, the posterior regions concentrate around the true parameters. \figfont{B,C}: Marginal posterior samples for $\lambda$ and $\mu$ respectively, along with the initial prior distributions for each (with prior means shown with dots on the curves). \figfont{A inset}: area of the HDP region as a function of $M$, showing $\sim M^{-1}$ scaling.}
    \label{fig:inferencefig}
\end{figure}

The results of performing Bayesian inference on synthetic data can be seen in \cref{fig:inferencefig}. The data considered here is the same as that in \cref{fig:logL}, with fixed parameters and varied number of snapshots. In panel \figfont{A}, we depict the 89\% HDPR (highest density posterior region) \cite{gelman1995bayesian}, computed by evaluating $L(\params \with \data)p(\params)$ for each posterior sample $\theta$, retaining the top 89\%, and then computing the convex hull of these samples. The resulting uncertainty qualitatively mimics the lesson seen from the likelihood. For small values of $M$, the posterior concentrates around the {parameter region that produces an expected number of particles equal to the observed count}, and then resolves to be sharply peaked around the true parameters with {data from more snapshots.} Qualitatively, the same result can be seen in the marginal posterior samples seen in \cref{fig:inferencefig}\figfont{B} and \figfont{C}. To more explicitly quantify the dependence of uncertainty on data, we plot the HDPR area as a function of the number of snapshots $M$ in the inset of \cref{fig:inferencefig}\figfont{A}, and find a convincing $\sim M^{-1}$ scaling. This scaling can be interpreted as arising from the product (an area) of the uncertainty of each parameter $\sim M^{-1/2}$ arising from the central limit theorem. 

\subsection{Cell-to-cell heterogeneities and their impact on inference}

\label{subsect:hetero}

In the previous figure, we assumed all cells were identical in shape and source location. However, in experimental setups, there is enormous cell-to-cell heterogeneity. The likelihood results shown in \cref{fig:logL} suggest that this heterogeneity should affect inference quality, as both domain size and source location both modulate the total information. However, those quantities were generated through the computation of the information for a single observation. How do these propagate to a heterogeneous population? To investigate this heterogeneity, we extend the information matrix calculation to now assume that the domain parameters $\domainparams=(z,L)$ follow some probability distribution. The expectation can be extended straightforwardly by the Law of Total Expectation: 
\begin{equation}
  I_{ij}(\params) = -  \mathbb{E}_{\params,\domainparams}[\partial_{\theta_i}\partial_{\theta_j} \log L(x;\params, \domainparams)] = - \mathbb{E}_{\domainparams} \mathbb{E}_{\params}[ \partial_{\theta_i}\partial_{\theta_j} \log L(x;\params| \domainparams)].
\end{equation}
Ultimately, this just reduces to integrating the results of previous figures against the probability distributions for $z$ and $L$. To parameterize the noise, we take $z/L\given L \sim \mathrm{Beta}(\sigma_z^{-1},\sigma_z^{-1})$  and $L \sim \Gamma(\sigma_L^{-2}, \sigma_L^2)$. The choice of distribution for source location was chosen so that $\EE_{\domainparams}[z\given L] = L/2$ and increasing $\sigma_z$ spreads out the source symmetrically around this location, with $\sigma_z=1$ corresponding to spatially uniform locations.  Similarly, $L$ is supported on the non-negative numbers, and the distribution is chosen so that $\EE_{\domainparams}[L]  = 1$ and increasing $\sigma_L$ spreads out the distribution of cell sizes in a plausible way. 
 
\begin{figure}[htb]
    \centering
    \includegraphics[width=0.9\textwidth]{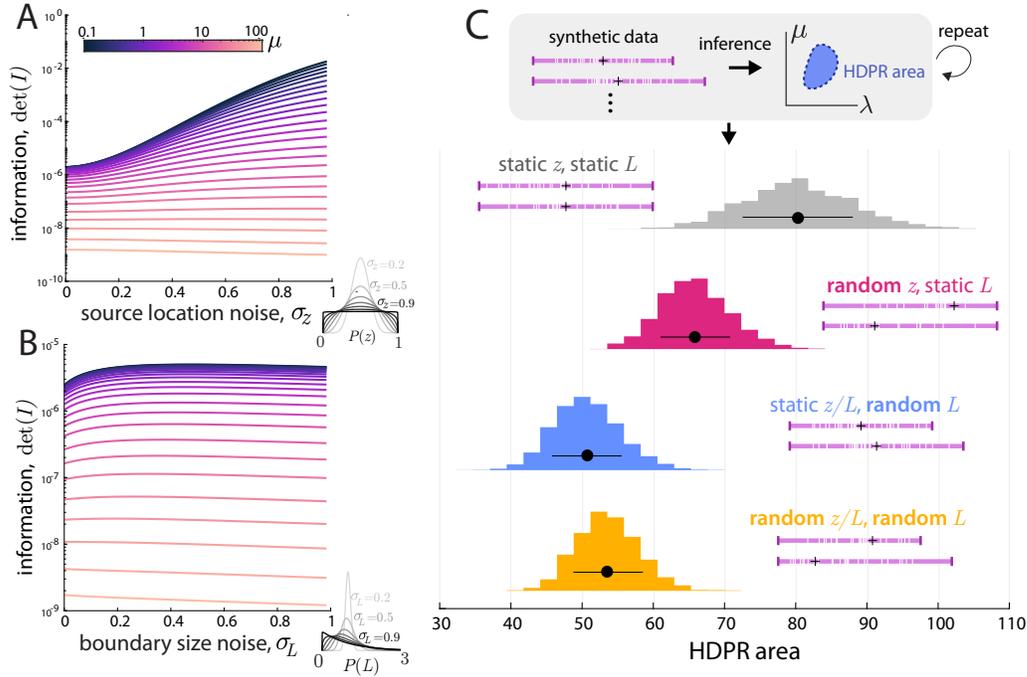}
    \caption{\textbf{Impact of cell-to-cell heterogeneities.} \figfont{A}: Information with $L=1$, $\lambda=500$ fixed and source location varied by a spatially symmetric parameterization with  $\sigma_z \to 0$ is a fixed source and $\sigma_z=1$ is spatially uniform. For low values of $\mu$, spatial heterogeneity \textit{increases} information, whereas it \textit{decreases} the information for large $\mu$. \figfont{B}: Information with $z=L/2$ and varied $L$ parameterized by a distribution such that $\EE_{\domainparams}[L]=1$ and increasing $\sigma_L$ spreads out the distribution. Again, lower $\mu$ scenarios have increased information from spatial heterogeneity. \figfont{C}: Repeated Bayesian inference on synthetic data with $\mu=1$, varied source location $\sigma_z=0.75$, cell size $\sigma_L=0.25$, and both random. All heterogeneous setups have lower uncertainty (smaller HDPR area) than static $z,L$.  }
    \label{fig:random_areas}
\end{figure}

The impact of cell-to-cell heterogeneity on the information can be seen in \cref{fig:random_areas}. Panel \figfont{A} the information $\det I$ as a function of $\sigma_z$ for various values of $\mu$ and fixed $\lambda=500, L=1$. For low values of $\mu$, varying source locations \textit{increase} the amount of information in observations, whereas for large $\mu$, they lower the information. This is in alignment with the results seen in \cref{fig:info}  for fixed source locations, where low values of $\mu$ corresponded to non-trivial optimal source locations, and large $\mu$ have optimal source locations in the center. 

A similar effect can be seen for varying $L$ in panel \figfont{B} of \cref{fig:random_areas}: for low values of $\mu$, varied cell sizes increase the information, whereas for large $\mu$, they decrease it. This result is somewhat more surprising. For increasing values of $\sigma_L$, the median $L$ decreases, meaning that more mass of the distribution is shifted toward cell sizes \textit{smaller} than the typical size. The fixed-size results say that smaller cells encode less information. However, taken together, the population-heterogeneous results suggest that the information from large cells offsets the lack thereof from small, and creates a net gain in information. 

An outstanding curiosity remains of whether these information results are meaningful for actual inference. To investigate this, we repeat the Bayesian inference pipeline from the previous section for various setups of heterogeneities: varied $z$, varied $L$, and both varied. For a fixed experiment of $M=100$ images, synthetic data (with $\mu$ chosen in the regime of increase information) is generated and inference is performed. The HDPR area is computed as a measure of the overall uncertainty in the estimation. This inference experiment is then repeated 1000 times, and the results can be seen in \cref{fig:random_areas} panel \figfont{C}. The results from estimation support the previous: heterogeneity does indeed \textit{increase} the amount of information in the snapshots, and \textit{lowers} uncertainty in estimation. All three scenarios with random source location, cell size, or both random yield smaller HDPR areas than the static scenario.

\section{Support for the inference framework.}

\subsection{Proof of Theorem \ref{thm:main}}
\label{sec:support:thm}

In order to prove that the locations of particles is a Poisson spatial process, we need only show that for any finite collection of disjoint sets $A_1, \ldots, A_k$, the number of particles in the sets are independent and Poisson distributed. If we find a function $u$ such that $\int_A u(x) \dd x = P(A)$ for all $A$, then this the intensity measure.

\smallskip 

\noindent \textit{Poisson distribution.} To fix notation, let $\rho_A(t)$ be the probability that a particle that is distributed according to $\f$ at time 0 and is in the set $A$ at time $t > 0$. Recall, however, that in the model the particles emerge at a rate $\lambda$ according to a Poisson arrival process for some interval of time $(-h,0)$ and we are counting the particles still in the system at time zero. Since this means that the initial times are iid and uniformly distributed over $(-h,0)$, we can define 
\begin{equation}
    \rho^{(h)}_A = \frac{1}{h} \int_{-h}^0 \rho_A(-s) \dd s = \frac{1}{h} \int_0^h \rho_A(t) \dd t.
\end{equation}
By Poisson thinning, if the total number of particles that emerge is $\text{Pois}(\lambda h)$, then the distribution of particles in $A$ at time zero is $\text{Pois}(\lambda h \rho_A^{(h)})$. Taking the limit as $h \to \infty$, we have that the number of particles in $A$ has a Poisson distribution with mean
\begin{equation} \label{eq:poisson-mean}
    \lambda \rho_A = \lambda \int_0^\infty \rho_A(t) \dd t.
\end{equation}

\smallskip 

\noindent \textit{Intensity measure.} The time-dependent PDE that describes the probability density of particles that have not yet been removed by state-switching or by being absorbed at the boundary of a domain is \cite[p.~230]{cox1977theory}
\begin{equation} \label{eq:pde-rho}
\begin{aligned}
    \partial_t \rho(x,t) &= L^* \rho(x,t) - \mu \rho(x,t), &x \in \domain, t > 0, \\
    \rho(x,t) &= 0, & x \in \partial \domain, t > 0, \\
    \rho(x,0) &= \f(x), &x \in \domain.
\end{aligned}
\end{equation}
Defining %
$\rho(x) = \int_0^\infty \rho(x,t) \dd x$, \cite[p.~163-167]{gardiner2004handbook} we integrate \eqref{eq:pde-rho} to find
\begin{equation}
\begin{aligned}
    \big(\lim_{t \to \infty} \rho(x,t)\big) - \rho(x,0) &= L^* \rho(x) - \mu \rho(x), &x \in \domain, \\
    \rho(x) &= 0, & x \in \partial \domain.
\end{aligned}
\end{equation}
Since particles exit the domain in finite time with probability one, $\rho(x,t) \to 0$ as $t \to \infty$. Imposing the initial condition yields
\begin{equation}
\begin{aligned}
    L^* \rho(x) - \mu \rho(x) &= - \f(x), & x \in \domain \\
    \rho(x) &= 0, & x \in \partial \domain.
\end{aligned}
\end{equation}
In order to achieve the desired mean for the Poisson distribution in \eqref{eq:poisson-mean}, we define $u(x) = \lambda \rho(x)$ and note that $u$ can be attained by multiplying the initial distribution function by $\lambda$. This yields the BVP stated in the theorem and the proof is complete.

\hfill $\square$

\subsection{Nondimensionalization}
\label{sec:support:nondim}

In this subsection, we outline the nondimensionalization argument justifying the SDE  \eqref{eq:defn-X}. Suppose that the original SDE satisfies the SDE\begin{equation} \label{eq:defn-X-dim}
    \dd X(t) = \alpha(X(t)) \dd t + \sqrt{2D} \dd W(t)
\end{equation}
with initial condition $X(t_0) = x_0$. The particle can meet one of two fates: a state-switch due to degradation, at rate $\kswitch$, or exit through the boundary of the domain. Define
\begin{displaymath}
    \tilde t \coloneqq \frac{D}{L_0^2} t, \,\, \tilde x \coloneqq \frac{1}{L_0} x, \text{ and } \tilde{\alpha}(\tilde x) := \frac{L_0}{D} \alpha (x).
\end{displaymath}
Then we claim that $\tilde X(\tilde t) \coloneqq \frac{1}{L} X(t)$ satisfies 
\begin{equation} \label{eq:defn-X-tilde}
    \dd \tilde{X}(\tilde{t}) = \tilde{\alpha}(\tilde{X}(t)) \dd \tilde{t} + \sqrt{2} \dd \widetilde{B}(\tilde{t})
\end{equation}
where $\tilde{B}(\tilde{t})$ is a standard Brownian motion. The argument follows from the integral form of  the SDE:
\begin{displaymath}
\begin{aligned}
    \tilde{X}(\tilde{t}) - \tilde{X}(\tilde{t}_0) &= \frac{1}{L_0} \left(X(t) - X(t_0)\right) \\
    &= \int_{t_0}^t \frac{1}{L_0} \alpha(X(s)) \dd s + \sqrt{\frac{2D}{L_0^2}} \left(B(t) - B(t_0)\right).
\end{aligned}    
\end{displaymath}
For the drift term, under the substitution $\tilde s = \frac{D}{L_0^2} s$ we have
\begin{displaymath}
\begin{aligned}
    \int_{t_0}^t \frac{1}{L_0} \alpha(X(s)) \dd s &= \int_{\frac{D t_0}{L_0^2}}^\frac{D t}{ L_0^2} \frac{L_0}{D} \alpha\left(X\Big(\frac{L_0^2}{D} \tilde s\Big)\right) \dd \tilde s \\
    &= \int_{\tilde{t_0}}^{\tilde{t}} \frac{L_0}{D} \alpha\left(L_0 \tilde{X} (\tilde{s})\right) \dd \tilde s \\
    &= \int_{\tilde{t}_0}^{\tilde{t}} \tilde{\alpha}\left(\tilde{X} (\tilde{s})\right) \dd \tilde s
\end{aligned}
\end{displaymath}
Meanwhile, define $\tilde{B}(\tilde t) = \sqrt{\frac{D}{L_0^2}} B(t)$. Then $E\big(\tilde{B}(\tilde{t})\big) = 0$ and 
\begin{equation}
    \textrm{Var}\big(\tilde{B}(\tilde t)\big) = \frac{D}{L_0^2} \textrm{Var}(B(t)) = \frac{D}{L_0^2} t = \tilde t
\end{equation}
meaning that $\tilde{B}$ is a standard Brownian motion.

In terms of the essential stopping times, we define $\mu = L_0^2/D$ (as in \eqref{eq:defn-lambda-mu}), and so $\tilde{\tau}_{\mathrm{switch}} = \tauswitch D/L_0^2$ is the corresponding nondimensional switch time. Moreover, defining $\tilde{\domain} = \{x \suchthat L_0 x \in \domain\}$ and $\tilde{\tau}_\mathrm{exit} := \inf \{\tilde t > \tilde{t}_0 \suchthat \tilde{X}(\tilde t) \notin \tilde{\domain}\}$. We can readily show path-by-path that $\tilde{\tau}_\mathrm{exit} = \tauexit D/ L_0^2 $. Analogous to the dimensional case, we define $\tilde{\tau} \coloneqq \min(\tilde{\tau}_\mathrm{switch}, \tilde{\tau}_\mathrm{switch})$

Altogether, it follows that, for any $A \in \domain$, if we let $\tilde{A} = \{x \suchthat L_0 x \in \domain\}$, then
\begin{equation}
    P(\tilde{X}(0 \wedge \tilde{\tau}) \in \tilde{A}) = P(X(0 \wedge \tau) \in A)).
\end{equation}
Hence, if we wish to solve for the intensity measure of a given dimensional system, we can rescale space by $L_0$ and the rates as in \eqref{eq:defn-lambda-mu} and use the intensity measure of the associated nondimensionalized system for statistical inference.

\subsection{Derivation of the likelihood function}
\label{sec:support:likelihood}

\begin{proposition}
 Let $\{\mathcal{A}_K\}_{K = 1}^\infty$ (with $\mathcal{A}_K = (A_{1,K}, \ldots, A_{K,K})$) be a sequence of $\domain$-partitions  that has the following properties:
\begin{enumerate}[(i)]
    \item the partitions are nested in the sense that if $K_1 < K_2$, then every set in $\mathcal{A}_{K_2}$ is contained within some member of $\mathcal{A}_{K_1}$; and 
    \item the partition mesh sizes, $\big(\max_{k \in \{1, \ldots, K\}}\{|A_{k,K}|\}\big)$, where $|A|$ denotes the volume of a set $A$, decrease to zero with $K$. 
\end{enumerate}

Fix $\vec{x} = \big(x_1, x_2, \ldots, x_{n(\vec{x})}\big)$ to be a single snapshot of particle locations arising from the dynamics defined for \cref{thm:main} and let $\pi(\params)$ be a prior distribution for $\params = (\lambda, \mu)$. Define $\pi_K(\params, \vec{x})$ be the posterior distribution of $\params$ under the prior $\pi$ when the locations are reported through binning in the sets of the partition $A_K$. Then,
\begin{equation}\label{eq:limiting-posterior}
    \pi(\theta \given \vec{x}) \coloneqq \lim_{K \to \infty} \pi_K(\theta \given \vec{x}) = \frac{L(\theta \with \vec{x}) \pi(\params)}{\int_\Theta L(\vartheta \with \vec{x}) \pi(\vartheta) \dd \vartheta}.
\end{equation}
where $L(\params \with \vec{x})$ is the likelihood function given in \cref{defn:likelihood}.
\end{proposition}

\begin{remark} A construction of such a sequence of partitions can be achieved, for example, by starting with any partition of $\domain$ and creating each refinement by randomly selecting one of the partition's sets and splitting in half.
\end{remark}

\begin{proof} We can represent the collection of particle locations as a measure $\delta_{\vec{x}}(x) \coloneqq \sum_{i=1}^{n(\vec{x})} \delta_{x_i}(x)$ \cite{daley2003introduction} and for each member $A_{k,K}$ of a partition $\mathcal{A}_K$, we can define
\begin{equation}
    n_{k,K}(\xvec) \coloneqq \int_{A_{k,K}} \delta_{\vec{x}}(x) \dd x.
\end{equation}
We then can write the posterior distribution 
$\pi_K(\cdot | \vec{x})$ associated with the partition $\mathcal{A}_K$ as follows:
\begin{equation} \label{eq:posterior-explicit}
    \pi_K(\theta \given \vec{x}) = \frac{P_{\theta} \big(N(A_{1,K}) = n_{1,K}(\xvec), \ldots, N(A_{K,K}) = n_{K,K}(\xvec)\big)\pi(\theta)}{\int_{\Theta} P_\vartheta\big(N(A_{1,K}) = n_{1,K}(\xvec), \ldots, N(A_{K,K}) = n_{K,K}(\xvec)) \pi(\vartheta) \dd \vartheta}
\end{equation}
Here the subscript of the probability measure denotes that we are evaluating the probability with the indicated parameter vector. The full parameter space is denoted $\Theta$, which in our case is $\rbb_+ \times \rbb_+$.
Since the probability in the above formula is expressed in terms of an event comprised of independent Poisson random variable results, we have the following explicit formula for the likelihood (suppressing dependence on $K$), which follows from Remark \ref{rem:poisson}: 
\begin{equation} \label{eq:likelihood-A-finite}
\begin{aligned}
    P_\theta\big(N(A_1) = n_1(\xvec), &\ldots, N(A_K) = n_K(\xvec)) \\
    &= \prod_{k = 1}^K e^{-\int_{A_k} \!\! u(x \with \theta) \dd x} \Big(\int_{A_k} \!\! u(x \with \theta) \dd x\Big)^{n_k(\vec{x})} \frac{1}{n_k(\vec{x})!}.
\end{aligned}
\end{equation}

Now, because the partitions are nested, there exists a $K_*$ such that for all $K > K_*$, all partition sets contain at most one particle. If we write $A_{k_i}$ to be the bin containing the particle location $x_i$, we have that $n_{k_i}(\vec{x}) = 1$ for each $i$, and $n_k(\vec{x}) = 0$ otherwise. For the purpose of writing the likelihood function, we gather all bins not containing a particle into one set 
\begin{displaymath}
    A_0 := \domain \setminus \Big\{\bigcup_{i=1}^{n(\vec{x})} A_{k_i}\Big\}.
\end{displaymath}
For such a partition, \eqref{eq:likelihood-A-finite} has the form
\begin{equation} \label{eq:likelihood-A-partition}
\begin{aligned}
    & P_\theta\big(N(A_1) = n_1(\xvec), \ldots, N(A_K) = n_K(\xvec)) \\
    & \qquad \qquad = P_\theta\big(N(A_0) = 0, N(A_{k_1}) = 1, \ldots, N(A_{k_n(\vec{x}}) = 1\big) \\
    & \qquad \qquad = P_\theta\big(N(A_0) = 0\big) \prod_{i = 1}^{n(\vec{x})} P_\theta\big(N(A_{k_i}) = 1\big) \\
    &\qquad \qquad = e^{-\int_{A_0} \!\! u(x \with \theta) \dd x} \prod_{i = 1}^{n(\vec{x})} e^{-\int_{A_{k_i}} \!\! u(x \with \theta) \dd x} \Big(\int_{A_{k_i}} \!\! u(x \with \theta) \dd x\Big) \\
    & \qquad \qquad = e^{-\int_\domain u(x \with \theta) \dd x} \prod_{i = 1}^{n(\vec{x})} \int_{A_{k_i}} \!\! u(x \with \theta) \dd x.
\end{aligned}
\end{equation}
Each term in the final product is going to zero, but by the Lebesgue Differentiation theorem (and observing that for the $i$th particle location, the sequence of $x_i$-containing partition sets $A_{k_i,K}$ is nested), we have that
\begin{displaymath}
    \lim_{K \to \infty} \frac{1}{|A_{k_i,K}|} \int_{A_{k_i}} \!\! u(x \with \theta) \dd x = u(x_i).
\end{displaymath}
It follows that 
\begin{equation}
\begin{aligned}
    \lim_{K \to \infty} \frac{P_\theta\big(N(A_{\ep,1}) = n_1(\xvec), \ldots, N(A_{\ep,K}) = n_K(\xvec)\big)}{\prod_{i=1}^{n(\vec{x})} |A_{k_i,K}|} = e^{-\int_\domain u(x) \dd x} \prod_{i=1}^{n(\vec{x})} u(x_i).
\end{aligned}
\end{equation}
The right-hand side is precisely $L(\theta \with \vec{x})$, the function appearing in the likelihood definition \eqref{eq:likelihood}.
Multiplying the numerator and denominator of \eqref{eq:posterior-explicit} by the rescaling factor $\prod_{i=1}^{n(\vec{x})} |A_{k_i,K}|^{-1}$, we obtain the limit of a sequence of posterior distributions $\pi_K(\cdot \given \vec{x})$, which can be written
\begin{displaymath}
    \pi(\theta \given \vec{x}) \coloneqq \lim_{K \to \infty} \pi_K(\theta \given \vec{x}) = \frac{L(\theta \with \vec{x}) \pi(\params)}{\int_\Theta L(\vartheta \with \vec{x}) \pi(\vartheta) \dd \vartheta}.
\end{displaymath}
This justifies the choice of $L$ appearing in \cref{defn:likelihood}.
\end{proof}

\subsection{Spatial binning encodes less statistical information}
\label{sec:supp:infobin}

The previous section provides a formal derivation of the point process likelihood \eqref{eq:likelihood} by considering increasingly fine-scale partitions of the spatial domain. In this subsection, we motivate this further by showing that the resulting likelihood encodes higher statistical information than any spatial binning of the process, stated as the following proposition. 
\begin{proposition}
Let $u(x;\theta)$ be the intensity of a Poisson point process on a domain $x\in\Omega$ with underlying parameter $\theta$. Let $\{A_j\}_{j=1}^{j=K}$ be any non-overlapping partition of the domain and denote observations of the binned process \begin{equation}
N_j = N(A_j)  \sim \mathrm{Poisson}\left( \int_{A_j} u(x) \dd x \right),  \qquad j=1,\ldots,K
\end{equation}
with corresponding log-likelihood $\ell_\mathrm{bin}(\theta \with \data)$. Let  $I_\mathrm{bin}(\theta) = -\EE_\theta [ \partial_{\theta \theta} \ell_\mathrm{bin}(\theta\with \data)] $ be the information for the binned process and $I_{\mathrm{unbin}}(\theta) = -\mathbb{E}_{\theta}[\partial_{\theta \theta} \log L(\theta \with \data)]$ be the point process information from likelihood \eqref{eq:likelihood}. Then, 
$$
I_\mathrm{bin}(\theta) \leq I_\mathrm{unbin}(\theta).
$$

\end{proposition}

\begin{proof}
The likelihood for the binned process is  this process is \eqref{eq:likelihood-A-finite}, and the log-likelihood is 
\begin{align}\ell_\mathrm{bin}(\theta \with \data)  &\coloneqq \log L_\mathrm{bin}(\theta \with \data) \\ &= -\int_{\domain} u(x \with \theta) \dd x + \sum_j n_j(\data) \log \int_{A_j} u(x \with \theta) \dd x  + \sum_j\log n_j(\data)!.
\end{align}
For a parameter $\theta$, \begin{equation}
\begin{aligned}
I_\mathrm{bin}(\theta) &= -\EE_\theta [ \partial_{\theta \theta} \ell_\mathrm{bin}(\theta\with \data)] \\
&= \int_{\domain} u_{\theta \theta}(x \with \theta) \dd x - \sum_j \EE[n_j]  \partial_{\theta \theta} \log \left ( \int_{A_j} u(x \with \theta) \dd x \right)\\
&= \int_{\domain} u_{\theta \theta}(x \with \theta) \dd x - \sum_j \left(\int_{A_j} u(x \with \theta) \dd x \right) \partial_{\theta \theta} \log \left ( \int_{A_j} u(x \with \theta) \dd x \right), \\
&= \sum_{j} \frac{\left( \int_{A_j} u_\theta(x \with \theta) \dd x \right)^2}{\int_{A_j} u(x \with \theta) \dd x}.
\end{aligned}
\end{equation}
where we have used $\partial_{\theta \theta} n_j =0 $ and $\EE_\theta[n_j] = \int_{A_j} u (x \with \theta) \dd x$. We can compare this directly to the unbinned information obtained via Campbell's formula \eqref{eq:Info}, restated as
\begin{equation}
I_{\mathrm{unbin}}(\theta) = \int_{\domain} u_{\theta \theta} (x \with \theta) \dd x -  \int_{\domain} u \partial_{\theta \theta} (x \with \theta) \log u (x \with \theta) \dd x  = \int_{\domain} \frac{u_\theta(x \with \theta )^2}{u(x \with \theta)}\dd x.
\end{equation}
To prove $I_{\mathrm{bin}} \leq I_{\mathrm{unbin}}$, it must be shown that
\begin{equation}
\sum_{j} \frac{\left( \int_{A_j} u_\theta (x \with \theta) \dd x \right)^2}{\int_{A_j} u(x \with \theta) \dd x} \leq \int_{\domain} \frac{u_\theta(x \with \theta)^2}{u(x \with \theta)}\dd x = \sum_{j} \int_{A_j} \frac{u_\theta(x \with \theta)^2}{u(x \with \theta)}\dd x,
\end{equation}
and noting this must hold for any $A_j$, this further reduces to showing
\begin{equation}
 \frac{\left( \int_{A_j} u_\theta(x \with \theta) \dd x \right)^2}{\int_{A_j} u(x \with \theta) \dd x} \leq \int_{A_j} \frac{u_\theta(x\with \theta)^2}{u(x \with \theta)}\dd x. \label{eq:inequal}
\end{equation}
Denote the inner product $\langle f, g \rangle_u \coloneqq \int  f(x) g(x) u(x)\dd x$. Then the assertion \eqref{eq:inequal} can be written  
\begin{equation}
\frac{\langle \partial_\theta \log u,1 \rangle_u^2  }{\langle 1, 1\rangle_u} \leq \langle \partial_\theta \log u, \partial_\theta \log u \rangle_u. \label{eq:inequal_CS}
\end{equation}
In this form \eqref{eq:inequal} is seen to be an expression of the Cauchy-Schwarz inequality, %
and the claim $I_{\mathrm{bin}} \leq I_{\mathrm{unbin}}$ is proved. \end{proof}

Cauchy-Schwarz provides an interpretation of this result. Equality is achieved when $\partial_\theta \log u \propto \mathrm{constant}$, or when $\partial_\theta u \propto u$. Intuitively, this says that scale parameters  (such as $\lambda$ in the main text) are insensitive to binning, but parameters that modulate the spatial variation of the intensity magnify the information loss from binning. %

\section{Discussion}

\label{sect:discuss}

In summary, we have established new understanding in the theory and practice of inferring dynamics of spatially stochastic birth-death-diffusion processes in heterogeneous domains from particle snapshot data alone. This pursuit was motivated by technical challenges arising in the study of spatial transcriptomics imaging data, where observable data often consists of positions of a stochastic population of RNA molecules, at a single time (due to cell fixation), in a cell-specific geometry. Using an argument based on the occupation measure of particles, the principal result is a rigorous statistical connection between particle paths and their observation as a Poisson point process with intensity governed by a  birth-death-diffusion PDE with appropriate boundaries. This explicit statistical description enables inference that faithfully incorporates often-ignored cell-to-cell heterogeneities. Surprisingly, we find that variation in source location and domain size can increase the accuracy of inference on the underlying dynamics in the regime where particles are long-lived relative to their typical diffusive exit time from the domain. Our study has broader context and future directions in the areas of inverse problems and mathematical modeling of gene expression data.

As an inverse-problem, the most important aspect of the work is the formal derivation and justification for the point process likelihood. This non-trivial likelihood must describe a stochastic number of particles and their position, in contrast to previous inference studies with single particle paths \cite{bernstein2016analysis,falcao2020diffusion} or population positions with fixed number \cite{chen_solving_2021}. Our work clarifies how point processes are the natural resolution of this technical hurdle. Moreover, we show that the point process description encodes more statistical information than spatial binning used in other inference of stochastic reaction-diffusion systems \cite{dewar2010parameter}. We are far from the first to consider mechanistic models in partnership with spatial point processes in biological systems, e.g.,  \cite{plank2015spatial,ovaskainen2014general}. The distinction of these studies in the ecological literature is subtle but important: in our work, the point process is an emergent statistical description from a birth-death-diffusion process rather than an assumed statistical model for observations. Although the PDE-constraint on the intensity bares resemblance to classical PDE inverse problems \cite{aster2018parameter,flegg2020parameter}, we emphasize the key difference in the observation noise. Rather than observations of the PDE solution (often boundary data) plus additive noise, our inverse problem corresponds to the observations of stochastic particle positions that serve as a proxy for the PDE solution. We provide a preliminary investigation of the statistical properties of this inverse problem, but there are several avenues of promising investigation in discerning similarities and differences in this PDE inverse problem against their classical counterpart. 

On the computational side, this connection to other inverse problems suggests there are are likely improvements toward more efficient inference. Such improvements will likely be necessary for large-scale inference on now available tissue-scale data \cite{eng_transcriptome-scale_2019}. The primary bottleneck of our approach is each evaluation of the likelihood requires solving the PDE model over all observed domains. The simplicity of the 1D model allowed for an analytical solution to the PDE and straightforward MCMC sampling. However, these luxuries break down quickly for more realistic investigations. Therefore, one may be able to construct improved samplers seen in other gene model inference\cite{kilic2023monte}, or leverage evaluations of the likelihood in approximate techniques such as variational Bayesian methods \cite{bleiVariationalInferenceReview2017,NEURIPS2018_747c1bcc}. However, these approaches still suffer the bottleneck of each likelihood evaluation requiring the solution to the model over all observed domains. An alternative avenue to circumvent this may be the use of neural network or related approaches. Such techniques have served useful in inverse problems on stochastic chemical reactions \cite{jiang_neural_2021}, classical PDE inverse problems \cite{raissi2019physics}, and particle positions \cite{chen_solving_2021}, including with domain generalization \cite{Cao2023.02.28.530379}. Therefore, this seems like a promising avenue of study for this problem at the intersection of these topics. Lastly, a variety of investigations remain in adapting the zoo of techniques in the toolbelt for inverse problems to this setup, ranging from model selection \cite{nardini2021learning} to optimal experimental design \cite{ryan2016review}. 

The steady-state birth-death-diffusion model was chosen at a deliberate level of complexity to distill the salient features of the data: snapshots of fluctuating populations in a domain. However, there is considerable  work that must be done in incorporating known biological complexities before our framework can be used on experimental smFISH data. We acknowledge the context of a vast literature on forward modeling of various non-spatial aspects of stochastic gene expression e.g., \cite{paulsson2005models,singh2013quantifying,cao2020analytical,GorinPRE2020,karamched2023stochastic}. The use of these stochastic models in inference for transcriptomics data is also increasingly prevalent and sophisticated \cite{herbach2017inferring,gomez2017bayfish,luo2022bisc}. We hope that our investigate serves as a segue for the field to transition into spatially resolved models that can be linked to this already available data. A non-exhaustive list of interesting complexities that may be included in our model are as follows. The production of RNA for many genes is known to be bursty \cite{jones2018bursting,donovan2019live}, with short stochastic intermittent windows of production. A spatial bursting model may be tractable with generating function techniques  \cite{cottrell_stochastic_2012} that serve fruitful for analytical understanding of non-spatial models \cite{GorinPRE2020,gorin2022modeling}.
 On the geometric side, the model must be extended to biological datasets of 2D and 3D FISH imaging \cite{omerzuThreedimensionalAnalysisSingle2019}. The models can also be extended to multiple compartments (the nucleus, cytoplasm) and serve as a platform for investigating the nature of RNA export through different possible boundary conditions, e.g., Robin boundary conditions arising through stochastic gating of particles \cite{lawley2015new} that may approximate the complex dynamics of RNA export \cite{kohler2007exporting,azimi_agent-based_2014}. Lastly, the movement model could be extended to include subdiffusive motion seen \cite{cabal2006saga}, known to profoundly dictate spatial patterns, especially near boundaries \cite{holmes2019subdiffusive}. 

In summary, our work serves as a basis for more detailed investigations of inferring dynamics from spatial dynamics from snapshots of stochastic particle populations, especially those arising from spatial transcriptomics data. The pursuit of understanding this biological data could only be accomplished by distilling it to its essential components and constructing that reflects these complexities. The analytical tractability of the model allowed for us to carefully examine the role of cell-to-cell heterogeneities which are surely present in a wide array of biological datasets. This surprising insight of robustness in inference arising from noise contributes a new dimension to the broader theme in mathematical biology of how noise can sharpen signals \cite{hanggi2002stochastic}. In an era where phenomenological statistical analysis of biological data is commonplace, we hope that our study will inspire others to consider rigorous modeling of biological data to discover other surprises and insights.

\section{Acknowledgements}

FD acknowledges NIH support 1DP2GM149554. The work of RBL and SAM was supported by the Laboratory Directed Research and Development program at Sandia National Laboratories, a multimission laboratory managed and operated by National Technology and Engineering Solutions of Sandia LLC, a wholly owned subsidiary of Honeywell International Inc. for the U.S. Department of Energy's National Nuclear Security Administration under contract DE-NA0003525.

\newpage

\appendix

\section{Numerical implementation details}
\label{sect:numerics}
\subsection{Stochastic simulation algorithm for synthetic data}

\cref{alg:particleSim} describes the generation procedure for synthetic data used in the inference setups of \cref{fig:inferencefig,fig:random_areas}. Importantly, we simulate trajectories of the particle process, and do inference using the point process theory from \cref{thm:main}. The simulation procedure is exact except for a small $\Delta_t$-controlled approximation error in computing whether a Brownian path has exited the domain \cite{smith2019spatial}. In all simulations, we take $\Delta_t=10^{-5}$, several orders of magnitude smaller than any other timescale. 

\begin{algorithm}
\caption{Stochastic simulation to generate particle positions for synthetic data.}\label{alg:particleSim}
\begin{algorithmic}[1]
\Procedure{SimulateParticles}{$k_\text{birth}, k_\text{death}, D, z, L, T, \Delta_t$}
\State $N\text{births} \sim \operatorname{Poisson}(\ksource T)$
\For{$n=1,\ldots,N{\text{births}}$}
\State $T\text{birth}_n \sim \operatorname{unif}(0,T)$
\State $\tau\text{decay}_n \sim \operatorname{exp}(1/\kswitch)$
\State $T\text{decay}_n \gets T\text{birth}_n+\tau\text{decay}_n$
\If{$T\text{decay}_n>T$}\Comment{only simulate if particle is possibly alive $T$}
\State $x\gets z$, \, $t_n\gets T\text{birth}_n$  \Comment{birth particle at source location} 
\State $\text{survived}_n\gets 1$
\While{$t_n \leq T$}
\State $x\gets x+  \sqrt{2D\Delta_t} \cdot \texttt{randn}$ \Comment{Brownian step}
\State $t_n\gets t_n+\Delta_t$
\If{$x<0$\text{ or }$x>L$}   \Comment{if exits domain}
\State  $\text{survived}_n\gets 0$ \textbf{break} \Comment{instantly absorbed, stop simulating}  \EndIf
\EndWhile
\If{$\text{survived}_n$} \Comment{if alive at observation $T$}
\State \textbf{append} $x$ \Comment{store final position}
\EndIf
\EndIf
\EndFor
\State \textbf{return} $x_1,\ldots,x_N$
\EndProcedure
\end{algorithmic}
\end{algorithm}

\subsection{Metropolis-in-Gibbs sampling for Bayesian inference}

See \cref{subsect:inference} for an explanation of prior choices for $\lambda, \mu$. The choice of prior for $\lambda$ provides a (marginal) conjugate posterior that can be computed analytically, so a Gibbs-in-Metropolis MCMC approach is used and described in \cref{alg:MCMCalg}.

\begin{algorithm}
\caption{MCMC Gibbs-in-Metropolis sampling of posterior distribution.}\label{alg:MCMCalg}
\begin{algorithmic}[1]
\Procedure{PosteriorSampMCMC}{$X,\sigma,\alpha_\lambda,\beta_\lambda, \mu_\mathrm{max}$}
\State \textbf{init} $\lambda_0\sim\operatorname{Gamma}(\alpha_\lambda,\beta_\lambda)$, $\mu_0\sim \operatorname{unif}(0,\mu_\mathrm{max})$
\For{$s=1\ldots,S$}
\LComment{MH step for $\mu$ first}
\State $\mu_* \gets \mu + \sigma \cdot \texttt{randn}$
\While{$\mu_*<0$ \text { or } $\mu_*>\mu_\text{max}$}
 $\mu_* \gets \mu + \sigma \cdot \texttt{randn}$ \Comment{only accept proposals for $\mu$ in range of prior}
\EndWhile
\State $\pi_* \gets \ell( \lambda_{s-1},\mu_*,X)$, $\pi_{s-1} \gets \ell(\lambda_{s-1},\mu_{s-1}, X)$   \Comment{proposal \& current log-likelihood}

\State $p \gets \min\{\exp(\pi_*-\pi_{s-1}),1\}$
\If{$\texttt{rand}<p$} 
$\mu_s \gets \mu_*$ \Comment{accept proposal}
\Else \, 
$\mu_s \gets \mu_{s-1}$   \Comment{keep current}
\EndIf
\LComment{Gibbs sample $\lambda$ using Gamma conjugacy}
\State  $\alpha,\beta$ $\gets$ using prior $\alpha_\lambda,\beta_\lambda$, data $X$, and $\mu_s$ in \eqref{eq:gamma_conj_update}  \Comment{Gibbs step uses current value of $\mu$}
\State $\lambda_s \sim \operatorname{Gamma}(\alpha,\beta)$ 
\EndFor
\State \textbf{return} $\lambda_0,\ldots,\lambda_S,\quad \mu_0,\ldots,\mu_S$
\EndProcedure
\end{algorithmic}
\end{algorithm}

\begin{figure}[htb]
    \centering
    \includegraphics[width=0.6\textwidth]{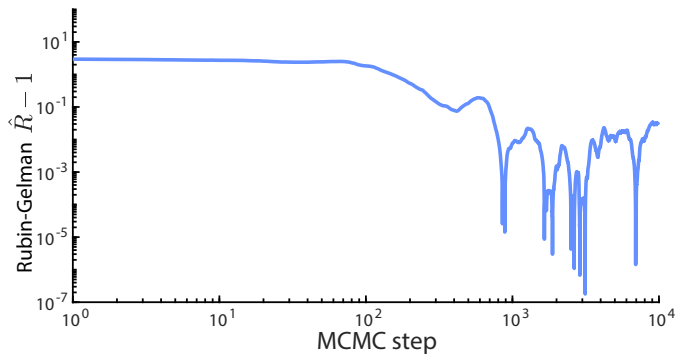}
    \caption{Rubin-Gelman $\hat{R}$-1 demonstrating convergence of MCMC (on the setup in Fig \ref{fig:inferencefig} with $M=10$) for two chains. Convergence is considered good when $\hat{R}-1<10^{-1}$,  achieved in approximately $10^3$ MCMC steps. }
    \label{fig:rhat}
\end{figure}

As is standard practice \cite{gelman1995bayesian}, we do sequence of warm-up phases to target the range of acceptance probabilities between $0.3$ and $0.5$, starting with $\sigma=1$ and rounds of 500 steps. After each warm-up round, if the acceptance frequency is below the target range, $\sigma \to \sigma/2$, and if above $\sigma \to 2\sigma$. Once in the range, the ending values of $\lambda,\mu$ are used as initial values and the full algorithm is run with this step size $\sigma$ for much greater number of steps. The values from the chains for this burn-in phase are discarded.  

To choose the number of MCMC steps, we investigate the Gelman-Rubin diagnostic $\hat{R}$ \cite{gelman1995bayesian} for the presumed worst-case scenario of $M=10$ snapshots in the setup of \cref{fig:inferencefig}. For $M=10$, we initialize multiple chains and compute $\hat{R}$ as a function of the number of MCMC steps, shown in \cref{fig:rhat}. Convergence is considered satisfactory for $\hat{R}<1.1$ \cite{gelman1995bayesian}, which is achieved around $10^3$ MCMC steps. We take $10^5$ MCMC steps for all setups in the main text and do not monitor convergence further.

\clearpage

\clearpage

\printbibliography

\end{document}